\documentclass{article}
\usepackage[utf8]{inputenc}
\usepackage{amsmath}
\usepackage{amssymb}
\usepackage{graphicx}
\usepackage{amsthm}
\usepackage{latexsym}
\newtheorem{theorem}{Theorem}

\usepackage{algpseudocode}
\usepackage{algorithmicx}
\usepackage{algorithm}

\newcommand{\beq}{\begin{equation}}
\newcommand{\eeq}{\end{equation}}
\newcommand{\barr}{\left[\begin{array}}
\newcommand{\earr}{\end{array}\right]}

\newcommand{\rank}{\mbox{rank}\,}
\newcommand{\fp}{\mathbb{F}_{p}}
\newcommand{\fq}{\mathbb{F}_{q}}

\newcommand{\bpf}{\begin{proof}}
\newcommand{\epf}{\end{proof}}

\newcommand{\zz}{\ensuremath{\mathbb{Z}}}

\newcommand{\ftwo}{\ensuremath{\mathbb{F}_{2}}}

\newcommand{\ff}{\ensuremath{\mathbb{F}}}
\newcommand{\al}{\alpha}

\title{Local Inversion of maps:\\Black box Cryptanalysis}
\author{Virendra Sule\\Department of Electrical Engineering\\
Indian Institute of Technology Bombay\\Mumbai 400076, India\\vrs@ee.iitb.ac.in}
\date{July 24, 2022}
\begin{document}

\maketitle

\noindent
Subject classification: cs.CR, cs.CC, math.NT

\begin{abstract}
    This paper is a short summery of results announced in a previous paper on a new universal method for Cryptanalysis which uses a Black Box linear algebra approach to computation of local inversion of nonlinear maps in finite fields. It is shown that one local inverse $x$ of the map equation $y=F(x)$ can be computed by using the minimal polynomial of the sequence $y(k)$ defined by iterates (or recursion) $y(k+1)=F(y(k))$ with $y(0)=y$ when the sequence is periodic. This is the only solution in the periodic orbit of the map $F$. Further, when the degree of the minimal polynomial is of polynomial order in number of bits of the input of $F$ (called low complexity case), the solution can be computed in polynomial time. The method of computation only uses the forward computations $F(y)$ for given $y$ which is why this is called a Black Box approach. Application of this approach is then shown for cryptanalysis of several maps arising in cryptographic primitives. It is shown how in the low complexity cases maps defined by block and stream ciphers can be inverted to find the symmetric key under known plaintext attack. Then it is shown how RSA map can be inverted to find the plaintext as well as an equivalent private key to break the RSA algorithm without factoring the modulus. Finally it is shown that the discrete log computation in finite field and elliptic curves can be formulated as a local inversion problem and the low complexity cases can be solved in polynomial time.  
\end{abstract}
\begin{description}
\item \emph{Keywords}\/: Cryptanalysis, Sequences over Finite Fields, Symmetric encryption, RSA, Discrete Logarithm.
\end{description}

\section{Introduction}
If $F:\ff^n\rightarrow\ff^n$ is a map acting in a cartesian space $\ff^n$ over a finite field $\ff$, then the \emph{Local Inversion Problem} of $F$ at a given $y$ in $\ff^n$ is the problem of computing all $x$ in $\ff^n$ such that
\beq\label{Eqfund}
y=F(x)
\eeq
Such a problem arises in most situations of Cryptanalysis of symmetric and public key primitives such as problem of key recovery (under known plaintext attack) in symmetric key algorithms, plaintext recovery in RSA without factoring the modulus, private key recovery in RSA without factoring the modulus or finding discrete logarithms in finite fields and elliptic curves. However in many such cases no algebraic or Boolean (in case of $\ff$ being $\ftwo$), model of the equation (\ref{Eqfund}) may be available or if available, is suitable for algorithmic solution due to a large number of unknowns in the equations. The map $F$ in such cases can only be practically used as a black box, an algorithm or a hardware machine or a computer code, available for forward computation $F(x)$ for a given $x$ efficiently. Also the number of times the forward computation $F(x)$ can be carried out should also be practically feasible. Hence it is suggested that the approach to solution of Cryptanalysis problems  using only forward computations by $F$ in practically feasible number of times, be called a \emph{Black Box approach} analogous to the well known Black Box Linear Algebra for solving linear systems of equations \cite{Wiedemann, VonzerGathen}. In this paper we develop such an approach and determine condition for local inversion to be computable in polynomial time in order $O(n^k)$. Although no assumption such as $F$ being invertible (or a permutation) is made, practically and most often, computing one solution $x$ of local inversion is a significant achievement although there may be more solutions.
This approach is thus a universal methodology to formulate diverse problems of Cryptanalysis of symmetric as well as public key cryptographic algorithms. This paper is a short summery of the theory of local inversion announced in \cite{suleLI} where a complete algorithm to determine all solutions of (\ref{Eqfund}) is also discussed. In general the theory presented in \cite{suleLI} shows that the problem of computing all solutions of (\ref{Eqfund}) is not solvable by purely forward computation alone. This is because solutions arise not just on periodic orbits of recurrences of the map $F$ but also on the chains $F^{(k)}(z)$, $k=0,1,2,\ldots$ where $z$ are points in the Garden of Eden (GOE) of $F$. The computation of GOE on the other hand is a NP-hard problem. This short paper is aimed at presenting computation of the restricted set of solutions $x$ of (\ref{Eqfund}) when the recurrence sequence $F^{k}(y)$ is periodic. There are good practical reasons why cryptographic algorithms result into maps $F$ which are one to one (or embedding) from the key space. Hence the equations to be solved in most practical cases come with an $F$ which is a permutation map. Hence all inverses $x$ for $y$ in their image are unique.   

\subsection{Local inversion using the Black Box approach under low complexity data}
It has been known in the literature that that the algebraic approach to solving (\ref{Eqfund}) is a very challenging computational problem. Either the algebraic equations expressing the relation (\ref{Eqfund}) involve too many (latent) variables and equations for which most approaches fail to scale up, or building of the algebraic (or Boolean model) of (\ref{Eqfund}) with minimum number of variables by symbolic elimination is another challenging computation. Hence it is worthwhile considering the Black Box approach in which $F$ is not specified algebraically but the output $y=F(x)$ can be computed for any input $x$ efficiently and moreover the number of times application of such a forward computation is needed is also practically feasible. It is shown in \cite{suleLI} that an existence of a periodic recurring sequence defined by $F$ and $y$ of polynomially bounded $O(n^k)$ \emph{Linear Complexity} (LC), is sufficient to solve one local inverse $x$ (or a rational root) in polynomial time\footnote{By polynomial time of a computation it is meant that the computation is achievable in $O(n^k)$ steps equivalently in $O((\log q)^k)$ steps in the field $\fq$}. This result has an important consequence to applications in Cryptanalysis. It turns out that Cryptanalysis of Symmetric key algorithms, RSA cryptanalysis of inverting ciphertetxt to plaintext as well as breaking all RSA encryptions by same public key from chosen ciphertext attack (CCA) and solutions of Discrete Logarithms in finite fields and Elliptic Curves (ECDLP) are all local inversion problems and hence the low LC of respective sequences is a sufficient condition for solving them in polynomial time. Local invesion is also a new methodology and a universal attack for solving all of these cryptanalysis problems. This article is a very short summery of the local inversion methodology announced in \cite{suleLI} which the reader is requested to refer for details and proofs of results. Inversion of maps for cryptanalysis has been studied in previous literature by the Time Memory Tradeoff (TMTO) attack \cite{Hellman,Hays,StampLow}. However TMTO does not consider finite fields structure on the domain of the maps. TMTO is also referred by name Rainbow attack \cite{Rainbow} which is governed by square root bound on the number of points in the image of $F$ which is of exponential size in the number of unknowns, hence TMTO attack is not feasible for realistic sizes of domains (or number of unknowns) of maps $F$ (or equation (\ref{Eqfund})).  

\section{Theory of local inversion of maps}
It is stated above that local inversion of a map $F$ at a value $y$ is equivalent to solving the roots of the system of (polynomial) equations (\ref{Eqfund}) if such a system is available. Hence this appears to be a problem solvable by Computer Algebra, by building the algebraic system of equations corresponding to the map $F$. Alternatively we observe that the solutions of (\ref{Eqfund}) generate trajectories of the dynamical system
\beq\label{DynSys}
y(k+1)=F(y(k)),y(k)\in\ff^n,k=0,1,2,\ldots
\eeq
which in turn are in one to one correspondence with the recurring sequences
\beq\label{Seq}
S(F,y)=\{y,F(y),F^(2)(y),\ldots\}
\eeq
Thus the theory behind local inversion of maps relates three different Mathematical domains, maps $F$ defined by algebraic equations in finite fields, dynamical systems defined by maps $F$ and the ensemble of sequences generated by maps $F$ with respect to points $y$ in their range. The complexity of solving the problem of local inversion depends on the distribution of linear complexities of sequences $S(F,y)$ where $y$ is taken over the range of $F$ when the sequences are periodic. The shift of focus from solving an algebraic system model of $y=F(x)$ to a Black Box model of $F$ generating sequence $S(F,y)$, gives a practically feasible solution to the local inversion problem under the special case of low LC of the sequence $S(F,y)$. However the solutions $x$ lying in the chains $F^{(k)}(z)$ for points $z$ in the GOE of $F$ are not solved by this method of computing minimal polynomials.

\subsection{Role of Linear Complexity (LC)}
In literature, LC was proposed as complexity measure of sequences in finite fields. Such complexity measure was used for qualifying the complexity of output sequences of maps defined by Block ciphers for counter inputs (in counter mode of operation), or for outputs of Stream ciphers \cite{Rueppel}. LC of several specially constructed sequences are also studied in Number Theory \cite{Nied}. Low LC $m$ of such sequences leads to the Berlekamp Massey attack which allows predicting the complete sequence over the whole period from $2m$ terms. In the local inversion approach however, the sequence $S(F,y)$ defined recursively by $F$ and $y$ is completely different from a sequence of an output stream of the cipher for counter input or output of Stream cipher for a fixed IV. To understand this difference consider an output sequence generated by an encryption function, $C(i)=E(K,Q(i))$ where $Q(i)$ is the internal state of a counter or a stream cipher. The key recovery problem for a known input $Q(0)$ and output $C(0)$ involves solving the equation $y=C(0)=E(K,Q(0))=F(K)$. Hence the recurrence sequence (\ref{Seq}) is completely different from $C(i)$. The LC of output sequences is used for modeling the sequence as an output sequence of LFSRs. But such modeling does not solve the symmetric key recovery problem or local inversion of $F$ at $y$. Recursive sequences such as $S(F,y)$ have not been studied previously for known cipher algorithms or general maps from the objective of inversion of maps. The application of the concepts of LC and minimal recurrence for local inversion is a fresh idea proposed in \cite{suleLI} and develops a Black Box approach for inversion of non-linear maps.

\subsection{A practical constraint in solving the problem}
Due to the finiteness of the field $\ff$ the recurring sequences $S(F,y)$ are quasiperiodic for any $y$. However even when the sequence is periodic the full sequence $S(F,y)$ is never practically available for computation since the period may be of exponential order in number of variables $n$ or the field size $n=\log |\ff|$. Hence only a small number of terms of the sequence $S(F,y)$ of order $O(n^3)$ are available for solving $x$ given $y$. Hence the data of limited number of terms of $S(F,y)$ may be insufficient to compute any solution $x$ or may compute a false positive solution which needs to eliminated by verifying whether $y=F(x)$. 

\subsection{Linear Complexity of a sequence}
Given a sequence $\hat{s}=\{s_0,s_1,s_2,\ldots\}$ over a finite field $\ff$, due to finiteness of $\ff$, the sequence is always ultimately periodic, i.e. there exist numbers $r\geq 0$ and $N>0$ such that
\beq\label{UPseq}
s_{(k)}=s_{(k+N)}\mbox{ for }k\geq r
\eeq
smallest $r$ is called the pre-period of $\hat{s}$ and smallest $N$ is called period. The sequence is called periodic of period $N$ when $r=0$. A polynomial can be associated to the periodic sequences satisfying relation (\ref{UPseq}) by defining a linear operator of shifting right by one step the sequence, 
\[
X(\hat{s})=\{s_{(k+1)},k=0,1,2,\ldots\}
\]
and scalar multiplication of $\hat{s}$ by constants. Then the relation (\ref{UPseq}) for periodic sequences can be expressed as
\[
s_k=X^{N}(s_k)\mbox{ for }k\geq 0.
\]
which is equivalent to
\[
(X^N-1)(\hat{s})=0
\]
In general if 
\[
p(X)=X^d-\sum_{i=0}^{(d-1)}\al_iX^i
\]
is any polynomial in $\ff[X]$ such that with the definition of $X$ as shift operator as above,
\beq\label{charpolyrelation}
p(X)(\hat{s})=0
\eeq
then $p(X)$ is called a characteristic polynomial of $\hat{s}$. Thus $(X^N-1)$ is a characteristic polynomial. A characteristic polynomial of smallest degree is unique and is called the \emph{minimal polynomial} of $\hat{s}$. The degree of the minimal polynomial is called the \emph{Linear Complexity} (LC) of the sequence. Now getting back to the sequence $S(F,y)$ defined in (\ref{Seq}), let $p(X)$ be a characteristic polynomial of $S(F,y)$. Then the relation (\ref{charpolyrelation}) signifies a \emph{linear recurrence relation} of degree $d$ satisfied by $S(F,y)$,
\beq\label{LRR}
F^{(d+j)}(y)-\sum_{i=0}^{(d-1)}\al_iF^{(i)}(y)=0
\eeq
for $j=0,1,2,\ldots$.
If $S(F,y)$ is periodic of period $N$, $(X^N-1)$ is always a characteristic polynomial. But there is a possibility of a lower degree characteristic polynomial satisfied by $S(F,y)$. Hence there exists a minimal polynomial for $S(F,y)$ which is a characteristic polynomial of least degree and is unique because the polynomial is monic. 

\subsection{Solution of the local inversion problem}
First consider the theoretical situation in which there are no limitations on the number of terms of the sequence $S(F,y)$ being available for inversion. In the first theorem below we show that computation of the minimal polynomial solves the local inversion problem when the sequence $S(F,y)$ is periodic. Let the minimal polynomial of a periodic $S(F,y)$ be denoted as
\beq\label{minpoly}
m(X)=X^m-\sum_{i=0}^{(m-1)}\al_{i}X^{i}
\eeq

\begin{theorem}
Let the sequence $S(F,y)$ be periodic of period $N$ and  $m(X)$ be its minimal polynomial. Then
\begin{enumerate}
    \item $m(0)\neq 0$.
    \item $N$ is the order of $m(X)$ over $\ff$.
    \item The solution of the local inverse of $y=F(x)$ in the periodic $S(F,y)$ is
    \beq\label{solution}
    x=(1/\al_0)[F^{(m-1)}(y)-\sum_{i=1}^{(m-1)}\al_iF^{(i-1)}(y)]
    \eeq
\end{enumerate}
\end{theorem}

\begin{proof}
Let $m(X)$ as denoted in (\ref{minpoly}) be the minimal polynomial of $S(F,y)$. Since the period of the sequence is $N$, $m(X)|(X^N-1)$ which implies $\al_0=m(0)\neq 0$ and since $N$ is the smallest number such that the recurrence 
\[
F^{(N)}(y)=y
\]
holds, $N$ is the order of $m(X)$. As $m(X)$ is also a characteristic polynomial of $S(F,y)$,
\[
m(X)(S(F,y))=0
\]
which is equivalent to
\[
F^{(m)}(y)-\sum_{i=0}^{(m-1)}\al_iF^{i}(y)=0
\]
Let $x$ be the solution of $y=F(x)$ in the same periodic orbit of the dynamical system (\ref{DynSys}). Then $S(F,x)$ is the sequence $S(F,y)$ shifted right by one index hence the sequence satisfies the same recurrence relation (\ref{LRR}) and the relation
\[
m(X)(S(F,x))=0
\]
which is equivalent to
\[
F^{(m)}(x)-\sum_{i=0}^{(m-1)}\al_iF^{i}(x)=0
\]
w.r.t. the minimal polynomial.
Substituting for $y=F(x)$ and solving for $x$ from the above equation as $\al_0\neq 0$ gives the solution as stated.
\end{proof}

It is important to observe that the solution $x$ of the local inverse is obtained in terms of a linear combination of the sequence $\{y,F(y),F^{(2)}(y),\ldots,F^{(m-1)}(y)\}$. 

\subsection{An incomplete algorithm}
We now come to the practically most relevant problem of local inversion of a map $F$ at $y$. The practical constraint to be faced is that the sequence $S(F,y)$ is not available for the full period because the period is exponential in $n$ (or close to the size of the field when $n=1$). However the partial sequence is available upto $M$ terms where $M$ is of polynomial size $O(n^k)$ for some $k$ as the sequence,
\beq\label{Mseq}
\{y,F(y),\ldots,F^{(M-1)}(y)\}
\eeq
In this case the recurrence relation (\ref{LRR}) is satisfied only upto the degree $m=\lfloor M/2\rfloor$ since there is no further data of the sequence available beyond $M$ terms. Hence the largest degree of the minimal polynomial is $m$. As shown in section 2.6 of \cite{suleLI} the minimal polynomial of the limited sequence upto $M$ terms is obtained from the unique solution of the linear system 
\beq\label{hankelsys}
H(k)\hat{\alpha}=h(k+1)
\eeq
where $H(k)$ is the Hankel matrix defined as
\[
H(k)=
\barr{llll}
y & F(y) & \ldots & F^{(k-1)}(y)\\
F(y) & F^{(2)}(y) & \ldots & F^{(k)}(y)\\
\vdots & \vdots & \ldots & \vdots\\
F^{(k-1)}(y) & F^{(k)}(y) & \ldots & F^{(2k-1)}(y)
\earr
\]
of largest rank for $k\leq\lfloor M/2\rfloor$, $\hat{\alpha}=(\al_0,\al_1,\ldots,\al_{(k-1)})^T$ the vector of co-efficients of the minimal polynomial of degree $k$ and the right hand side vector $h(k+1)$ is the last column of $H(k+1)$ after dropping the bottom entry.

\begin{algorithm}\label{IncompleteAlgo}
\caption{Incomplete algorithm to find the unique solution of $y=F(x)$ given $y$}
\begin{algorithmic}[1]
\Procedure{LocalInversion}{given $M$ terms of $S(F,y)$}
\State \textbf{Input} $M$ of $O(n^k)$ and the sequence (\ref{Mseq}).
\Repeat
\State set $m=\lfloor M/2\rfloor$.
    \If {$m=\rank H(m)=\rank H(m+1)$} 
    \State A unique minimal polynomial of degree $m$ exists
    \State Compute the minimal polynomial and solution $x$ as in (\ref{solution}).
        \If {solution satisfies $y=F(x)$}
        \State Return, ``Solution $x$".
        \State End procedure
        \Else 
        \If {$\rank H(m+1)>m$}
        \State Return ``Insuffienct data to compute the local inverse".
        \State End procedure
        \EndIf
        \EndIf
    \Else 
    \State $\rank H(m)<m$, Reduce $M\leftarrow M-1$ reset $m$
    \EndIf
\Until {$M=2$}
\State Return ``insufficient data to compute the local inverse".
\EndProcedure
\end{algorithmic}
\end{algorithm}

The incomplete algorithm is useful to decide whether the local inverse can be computed from the given data. Since the given data is the partial sequence of $M$ terms which is of polynomial size, the algorithm computes the inverse in polynomial time when the data is sufficient to compute the inverse. Algorithm also checks for false positive cases of solutions and if no solution is found returns that the data is insufficient. We can thus conclude 

\begin{theorem}
Let $M$ be of polynomial order $O(n^k)$. If the sequence (\ref{Mseq})
has minimal polynomial satisfying the recurrence relation (\ref{LRR}) of degree $m\leq\lfloor M/2\rfloor$ and the solution $x$ computed in (\ref{solution}) satisfies $y=F(x)$, then the local inverse is computable in polynomial time.
\end{theorem}
In practice the polynomial bound $O(n^k)$ is limited to $k=1,2,3$ for realistic cases.

\subsection{Embedding maps}\label{EmbeddingMaps}
The algorithm for local inversion described above is for maps $F:\ff^n\rightarrow\ff^n$ where the number of bits $n$ of the input and the output are same. In many practical situations of cryptanalysis however, the maps arise as $F:\ff^n\rightarrow\ff^m$ where $m>n$. Such a map is called an embedding map. The solution of local inverse $y=F(x)$ for such a map can be found by using the theory developed above for regular maps with $m=n$ as explained next.

\subsubsection{Solution of local inverse for embedding maps}
Consider a projection $\Pi$ from $\ff^m\rightarrow\ff^n$. For a vector $y$ in $\ff^m$, $\Pi(y)$ denotes a vector of some of the $n$ components of $y$. If $x$ is a local inverse of $y=F(x)$ the for any such projection, $\Pi(y)=(\Pi\circ F)(x)$ hence $x$ is also a local inverse of $y_1=\Pi(y)$ under the map $F_1=\Pi\circ F$. Since $F_1$ is a standard map in $\ff^n$ we can utilize the recurrence sequence generation using $F_1$ and the minimal polynomial to find a possible local inverse of $y_1=F_1(x)$. Then such a solution is verified with all other maps to verify whether $x$ satisfies $\Pi_k(y)=F_k(x)$ for all other projections $\Pi_k$ of $\ff^m$ to $\ff^n$. A solution is rejected if it fails to satisfy the equation for any $k$. Hence we can define the LC of an embedding map to be the smallest LC of the recurrences defined by $y_k=F_k(x)$ over all projections $\Pi_k:\ff^m\rightarrow\ff^n$. In fact it can be observed that a restricted set of projection on following subsets of co-ordinates of $y$ are sufficient to compute and verify local inverse $x$, 
\[
\begin{array}{lcl}
\Pi_1(y) & = & \{y^1,\ldots,y^n\}\\
\Pi_2(y) & = & \{y^2\ldots,y^{(n+1)}\}\\
\vdots & \vdots & \vdots\\
\Pi_{(m-n+1)} & = & \{y^{(m-n+1)},\ldots,y^m\}
\end{array}
\]
The computational effort depends on the minimal polynomial of recurrences $y_k=F_k(x)$ at the first choice of $k$ which chooses a projection. Hence an embedding map can still be locally inverted in polynomial time if there is an index $k$ such that the recurrence defined by $y_k=F_k(x)$ has LC of polynomial order in umber of the variables and the solution verifies with all other projections. 

\section{Typical maps in Cryptanalysis}
In this final section we outline some of the typical maps arising in Cryptography which can be considered for local inversion to show how local inversion approach is a uniform methodology for Cryptanalysis. However it must be understood that the local inversion is a theoretical methodology, its practical feasibility depends on how small is the LC of the local inversion problem as shown in Theorem 2. 

\subsection{Symmetric encryption algorithms}
The two types of algorithms are block ciphers and stream ciphers. 
\begin{enumerate}
    \item \textbf{Block ciphers}: An algorithm is described by $C=E(K,P)$ where $P,C$ are the blocks of plaintext and ciphertext respectively while $K$ is the symmetric key block. In the cases of Known Plaintext Attack (KPA) the pair $(P,C)$ of blocks is known to the attacker. Hence the local inversion problem is $y=F(x)$ where $y=C$, $x=K$ and $F(x)=E(x,P)$. The map $F$ is from $l$-bits of $K$ to block length in bits of $C$. In the chosen plaintext attack $P$ is chosen such that the inversion problem is easier. For local inversion $P$ can be chosen such that the LC of $F(x)$ is smallest or is polynomially bounded. Though such an input may not be easy to compute.
    \item \textbf{Stream ciphers}: The algorithm is described by the dynamical system with a state update map $x(k+1)=F(x(k))$ and an output map $y(k)=f(x(k))$. The initial condition $x(0)=(K,IV)$ consists of symmetric key $K$ (of bit length $l$) and an initialising vector $IV$. Hence the local inverse problem is defined by $y=\hat{F}(x)$ where $y$ is the vector of output stream $(y(k_0),y(k_0+1),\ldots,y_(k_0+l-1))$ while 
    \[
    \hat{F}(x)=((F^*)^{(k_0+i)}f(x,IV)),\;\mbox{for}\;i=0,1,2,\ldots,(k_0+l-1)
    \]
    The map $\hat{F}$ can be made embedding by collecting more samples of the output stream.
\end{enumerate}
In both cases of block and stream ciphers, the maps for key recovery are available in black box form with $l$-bits of input $x$ (key length) and output $y$ on which the inversion algorithm can be applied after choosing $M$ bounded by a polynomial order $O(l^k)$. 

\subsection{RSA cryptanalysis}
RSA has public keys $n=pq$ where $p,q$ are odd primes and an exponent $e$ such that $\gcd(e,\phi(n))=1$. The private keys are $p,q,d$ where $ed=1\mod\phi(n)$. Two local inversion problem can be defined as follows.
\begin{enumerate}
    \item Decryption of ciphertext: For the message $m\in[0,n-1]$ the ciphertext is $c=m^e\mod n$. Let $l$ be the bit length of $n$. For an $l$-bit number $x$ let $(x)$ denote the bit string in the binary expansion of $x$ and $[x]$ denote the operation of recovering the number $x$ from the binary string $(x)$. Then the map $F:\ftwo^l\rightarrow\ftwo^l$ is defined by $y=(c)$ and $F(x)=([x]^e\mod n)$. The sequence $S(F,y)$ in $[0,n-1]$ is
    \[
    \{c,c^e\mod n, c^{e^2}\mod n,c^{e^3}\mod n\ldots\}
    \]
    while the sequence $S(F,y)$ in $\ftwo^l$ is
    \[
    \{(c),([c]^{e}\mod n),([c]^{e^2}\mod n),([c]^{e^3}\mod n),\ldots\}
    \]
    It is well known that the sequence $c^{e^k}\mod n$ is periodic since $e$ is coprime to $\phi(n)$. The LC is obtained from the sequence of binary vectors. The message $m$ is solved as local inverse $(c)=F(x)$, $m=[x]$. This way of solving for $m$ shows how RSA ciphertext can be decrypted by local inversion without factoring $n$. In cases when $c$ causes low LC of the above sequence, the encryption can be broken in polynomial time as shown by Theorem 2. 
    \item Chosen Ciphertext Attack: In this attack a ciphertext $c$ is chosen by the attacker and the decrypted plaintext $m$ is sought as a verification. The relation is $m=c^d\mod n$. Hence the local inverse problem is defined by $y=m$ and $F(x)=c^x\mod n$. Since $d$ belongs to the ring $\zz_{\phi(n)}$ the number of bits of $x$ is equal to $\phi(n)$. Although $\phi(n)$ is not available to the attacker we can choose the number of bits to be $n$. The sequence $S(F,y)$ is
    \[
    \{m,c^m\mod n,c^{c^m\mod n}\mod n,\ldots\}
    \]
    while the map $F:\ftwo^l\rightarrow\ftwo^l$ is defined by $(m)=(c^{[x]}\mod n)$ and the sequence of binary vectors is
    \[
    \{(m),(c^{[m]}\mod n),(c^{[c^{[m]}\mod n]}\mod n),\ldots\}
    \]
    Note that the computation modulo $n$ automatically restricts the exponents modulo $\phi(n)$ hence the sequence generated is correct for local inversion of the map $F$. Periodicity of the sequences above are established in \cite{suleLI}. The sequence is converted to a sequence of binary expansions to compute the LC and the local inverse $x$. The inverse $x$ is converted back to a number in $[0,n-1]$. Now an important observation to be noted in this inversion is that the inverse of $m=F(x)$ is $x$ which is not necessarily $d$ but satisfies $ex=1\mod\phi(n)$. Hence although the private key $d$ is not exactly recovered, the inverse allows decryption  of any other ciphertext $\tilde{c}$ corresponding to the plaintext $\tilde{m}$ since
    \[
    (\tilde{c}^x\mod n)^e\mod n=\tilde{c}^{ex}\mod n=\tilde{c}
    \]
    hence $\tilde{m}=\tilde{c}^x\mod n$. This shows that local inversion by CCA on RSA breaks RSA equivalent to computing the private $d$ key without factoring the modulus $n$.
\end{enumerate}

\section{Discrete logarithm in finite fields and elliptic curves}
Discrete Logarithm problem (DLP) has been the first major tool for development of Public Key Cryptography as the Diffie-Hellman key exchange scheme (DH). Later the DH scheme was upgraded to elliptic curves leading to Elliptic Curve Cryptography (ECC). The DLP on Elliptic Curves (ECDLP) has not been found to have an algorithm for solution better than exponential complexity. What we show in this section is that both these problems can be addressed as local inversion problems. Hence it turns out that at least in cases when the map $F$ and given data $y$ are such that the sequence $S(F,y)$ has low LC, the DLPs can be solved efficiently. The situation in case of ECDLP is much complicated because the map $F$ turns out to be an embedding.

\subsection{DLP on finite fields}
In a prime field $\fp$ for a large prime $p$, the exponent function with a base $a$ in $\fp^*$, $\phi:[1,p-1]\rightarrow\fp,x\mapsto a^x\mod p$ is actually defined as a map from $\fp^*$ to $\fp^*$. Taking binary expansion of numbers in $[1,p-1]$ if $l$ is the bit length of $p$ this function expresses a map operation
\[
\begin{array}{llcl}
F:& \ftwo^l & \rightarrow & \ftwo^l\\
  & (x) & \mapsto & (a^{[x]}\mod p)\\
\end{array}
\]
where $(x)$ denotes the binary string corresponding to a number $x\in[1,p-1]$ and $[x]$ the reconstruction of the number $x$ from the binary string. Let $b$ be a given element in $\fp^*$ for a primitive element $a$. Then the equation for local inversion is $F((x))=(b)$. The map is available for black box computation. The sequence of binary vectors $S(F,(b))$ is
\beq\label{SeqDLFp}
\{(b),(a^{[b]}\mod p),(a^{a^{[b]}\mod p}\mod p),\ldots\}
\eeq
As defined by the map iterations $y(0)=(b)$, $y(k+1)=F(y(k))$ for $F$ defined as above on $\ftwo^l$. Reader is referred to \cite{suleLI} to see deatils justifying the periodicity of the sequence.

The solution of the DLP can thus be found as local inversion of this map $F$ at the value $(b)$. If the LC of the above sequence given upto polynomial number of terms $O(l^k)$ is small then as described in Theorem 2, the DLP can be solved in polynomial time using the Algorithm 1. For the solution of DLP over general field $\fq$ using local inversion the reader is referred to \cite{suleLI}.

\subsection{DLP on elliptic curves}
In the discrete log problem on an elliptic curve $E$ over $\fq$, there are given points $P$ and $Q=[m]P$ in $E$ where $m$ is the integral multiplier. It is required to solve for the multiplier $m$. Define the map
\[
F_{P}:m\mapsto[m]P
\]
then the local inversion of $Q=F_{P}(m)$ solves the ECDLP. However the map $F$ needs to be expressed in the standard form as before. We can do this by defining the map in $\ftwo^l$ for an appropriate $l$. 

\subsubsection{Formulation as local inversion}
The multiplier $m$ is less than the order of the cyclic group $n=<P>$ in $E$. Sometimes the group $E$ itself has prime order $n=\sharp E$ hence the order of $<P>$ is $n$. Hence to fix the number of bits $l$ in $m$ we consider estimates of the order of $E$. The well known bound on the order of $E$ is
\[
\sharp E(\fq)\leq q+1+2\sqrt{q}
\]
Assuming $q>4$ we have $\sharp E\leq 2q$. On the other hand the point $Q$ in $E$ has two co-ordinates in $\fq$. Thus the bit length of $\sharp E\leq 1+\log (2q)=2+\log q=l$ while the bit length of two co-ordinates of a point taken together is $1+\log (2q)=l$. Consider the map $F_{P}$ defining the scalar multiplication of $P$ in $E$ 
\[
\begin{array}{lcl}
F_{P} & : & \zz_n\rightarrow E\\
 & & m\mapsto [m]P
\end{array} 
\]
In the binary co-ordinate expansion on both sides this mapping is  
\beq\label{mapFPm}
F_{P}((m))=((Q_x),(Q_y))
\eeq
where $(m)$ denotes the co-efficients in the binary expansion of $m$ and $(Q_x)$, $(Q_y)$ are co-efficients in the binary expansions of co-ordinates of $Q=[m]P$. We shall denote the binary expansion of the co-ordinate pair of $Q$ as $(Q)$. 

\subsubsection{Formulation of $F_{P}$ as a map over the binary field}
In order to utilize the previous theory of inversion on the map (\ref{mapFPm}), it is necessary to express it as a map in the cartesian spaces of $\ftwo$ and understand whether it is an embedding. 

Let $r=1+\log n$ where $n$ is the order of $<P>$. Then $F_{P}$ in (\ref{mapFPm}) represents a map $F_{P}:\ftwo^r\rightarrow\ftwo^l$ where $r<l$. Thus $F_{P}$ is an embedding. Hence it is required to apply the theory of local inversion of embedding of section 2.5 to solve the embedding equation for the local inversion of $F_{P}(m)=Q$ from the binary representation in (\ref{mapFPm}). The application of Algorithm 1 requires that $n$ the order of $P$ is known. Let $t=l-r$, then the projection equations as refrred in the subsection (\ref{EmbeddingMaps}) give $(t+1)$ standard equations 
\[
\Pi_i\circ F_{P}(x)=\Pi_i((Q)),i=1,2,\ldots (t+1)
\]
denote by $F_i=\Pi_i\circ F_{P}$ and $y(i)=\Pi_i((Q))$ where $\Pi_i$ are projection on $r$ components of $y$. The details of this projection are described in \cite{suleLI}. Following Theorem 2 now we have

\begin{theorem}
If for any of the indices $i$, $1\leq i\leq (t+1)$ the projection equation $F_{i}(x)=y(i)$ has a periodic recurrence sequence $S(F_{i},y(i))$ with a LC $m\leq\lfloor M/2\rfloor$ where $M$ is of polynomial order $O(l^k)$ and the local inverse $x$ satisfies all other projection equations then the ECDLP is solved in polynomial time by Algorithm 1.
\end{theorem}

\section{Conclusions}
This brief paper summerizes the ideas on a new method which may be called as a Black Box approach to Cryptanalysis by Local Inversion of maps in finite fields. The methodology is universal in the sense that it is applicable to diverse problems of cryptanalysis of both symmetric as well as public key cryrptography. The most important conclusion arising from this method is that, since the recursive sequences associated with the inversion problems identified by the method have not been studied in the past for their LC for most of the maps $F$ in Cryptography, urgent work should be carried out to determine density of points $y$ in the image of $F$ for which the sequences $S(F,y)$ have low LC. Such low LC cases can arise due to plaintext inputs, IVs and parameters like primes and co-effcients of elliptic curves. Low LC is disruptive to the security of the algorithm. Hence these densities should be useful for deciding grading of security of cryptographic primitives and should be included as part of standards security for design of ciphers.

\end{document}